\documentclass[11pt,a4paper]{article}
\usepackage{amssymb}
\usepackage{amsmath}
\usepackage{amsfonts}
\usepackage{bbm}
\usepackage{amsthm}
\usepackage{mathrsfs}
\usepackage{hyperref}
\usepackage{color}
\usepackage[utf8]{inputenc}
\usepackage{stmaryrd}
\usepackage{xfrac}
\usepackage{graphicx}
\usepackage{url}

\definecolor{darkred}{rgb}{0.8,0.1,0.1}
\hypersetup{
     colorlinks=false,         
     linkcolor=darkred,
     citecolor=blue,
}

\usepackage{comment}
\usepackage[margin=2.6cm]{geometry}
\usepackage[all,cmtip]{xy}

\usepackage{marvosym}

\theoremstyle{plain}
\newtheorem{theo}{Theorem}[section]

\theoremstyle{definition}

\newtheorem{rem}[theo]{Remark}

\numberwithin{equation}{section}

\def\Diff{\mathrm{Diff}}

\def\Man{\mathsf{Man}_{2}}

\def\Ab{\mathsf{Ab}}
\def\CAlg{C^\ast\mathsf{Alg}}
\def\Loc{\mathsf{Loc}}

\def\dd{\mathrm{d}}
\def\op{\mathrm{op}}
\def\id{\mathrm{id}}
\def\cc{\mathrm{c}}

\def\OO{\mathcal{O}}
\def\AA{\mathcal{A}}

\def\bbZ{\mathbb{Z}}

\def\bbR{\mathbb{R}}
\def\bbC{\mathbb{C}}
\def\bbT{\mathbb{T}}
\def\bbS{\mathbb{S}}
\def\1{\mathbbm{1}}

\def\sk{\vspace{2mm}}

\newcommand{\ips}[2]{\langle #1,#2\rangle}
\newcommand{\bol}[1]{\mathbf{#1}}

\title{%
Non-existence of natural states for Abelian Chern-Simons theory
}

\author{%
Claudio Dappiaggi$^{1,a}$, \ Simone Murro$^{2,b}$ \ and \  Alexander Schenkel$^{3,c}$ \vspace{4mm}\\
{\small $^1$ Dipartimento di Fisica, Universit{\`a} di Pavia \& INFN, sezione di Pavia,}\\ 
{\small  Via Bassi 6, 27100 Pavia, Italy.}\vspace{2mm}\\
{\small $^2$ Fakult\"at f\"ur Mathematik, Universit\"at Regensburg, 93040 Regensburg, Germany.}\vspace{2mm}\\
{\small $^3$ School of Mathematical Sciences, University of Nottingham,}\\
{\small University Park, Nottingham NG7 2RD, United Kingdom.}\vspace{3mm}\\
{\footnotesize $^a$\texttt{claudio.dappiaggi@unipv.it}, $^b$\texttt{simone.murro@ur.de}, $^c$\texttt{alexander.schenkel@nottingham.ac.uk}}
 }

\date{February 2017}

\begin{document}

\maketitle

\begin{abstract}
We give an elementary proof that Abelian Chern-Simons theory, described as a functor from oriented surfaces to $C^\ast$-algebras, does not admit a natural state. Non-existence of natural states is thus not only a phenomenon of quantum field theories on Lorentzian manifolds, but also of topological quantum field theories formulated in the algebraic approach.
\end{abstract}

\section{\label{sec:intro}Introduction and summary}
A locally covariant quantum field theory (LCQFT) \cite{BFV} is a functor
$\AA :\Loc\to \CAlg$ from a category of (globally hyperbolic Lorentzian) spacetimes
to the category of $C^\ast$-algebras satisfying suitable physical axioms. 
The $C^\ast$-algebra $\AA(M)$ assigned to a spacetime $M$ is interpreted as the
algebra of quantum observables which can be measured in $M$. The $C^\ast$-algebra
homomorphism $\AA(f) : \AA(M)\to \AA(M^\prime)$ assigned to a spacetime embedding
$f:M\to M^\prime$ allows us to associate observables in larger spacetimes starting from observables in smaller ones.
\sk

For a quantum physical interpretation of a LCQFT $\AA :\Loc\to \CAlg$ it is 
necessary to choose for each spacetime $M$ a state  $\omega_M : \AA(M)\to \bbC$, 
i.e.\ a positive, normalized and continuous linear functional. The GNS-representation then leads to the 
usual Hilbert space formulation of quantum physics.
Motivated by the functorial structure of LCQFT, 
it seems natural to demand that the family of states $\{\omega_M\}_{M\in\Loc}$ 
is compatible with the functor $\AA :\Loc\to \CAlg$ in the sense that
\begin{flalign}
\omega_{M^\prime}\circ \AA(f) = \omega_M~,
\end{flalign}
for all $\Loc$-morphisms $f:M\to M^\prime$. Such compatible 
families of states are called {\em natural states} on $\AA: \Loc\to \CAlg$.
\sk

Even though the idea of natural states is very beautiful and appealing, there are hard obstructions 
to the existence of natural states. Early arguments
were already given by Brunetti, Fredenhagen and Verch in \cite{BFV}. Later, a no-go theorem
on the existence of natural states (under some additional assumptions) has been proven 
by Fewster and Verch in \cite{FV}. This no-go theorem makes use of very particular
properties of dynamical quantum field theories on Lorentzian spacetimes,
e.g.\ the concept of relative Cauchy evolution. 
\sk

As a consequence, it is not clear whether such no-go result holds true also for 
other classes of quantum field theories, such as topological ones, 
since these do not satisfy the hypotheses of \cite{FV}. 
In this short paper we address this question by considering Abelian Chern-Simons theory and 
proving that it does not admit a natural state. This result shows in particular that non-existence 
of natural states is not directly linked to quantum field theories on Lorentzian manifolds, 
as the proof of the no-go theorem in  \cite{FV} may suggest. 
It is a more general and broader phenomenon which may also occur in 
other kinds of quantum field theories, including the class of 
topological quantum field theories.
\sk

The outline for the remainder of the paper is as follows:
In Section \ref{sec:CStheory} we define the
Abelian Chern-Simons theory functor $\AA : \Man\to\CAlg$ by 
adapting the construction \cite{BDHS} of Abelian Yang-Mills theory.
This functor assigns $C^\ast$-algebras to $3$-dimensional manifolds of the form $\bbR\times \Sigma$,
where $\Sigma$ is a $2$-dimensional oriented manifold.
In Section \ref{sec:compact} we analyze the induced representation 
of the orientation preserving diffeomorphism group $\Diff^+(\Sigma)$ on
$\AA(\Sigma)$ and relate it to a representation of the mapping class group and, for
compact $\Sigma$, also to a representation of a suitable discrete symplectic group.
As a by-product, we show that there exists a $\Diff^+(\Sigma)$-invariant state 
on $\AA(\Sigma)$ for all compact $\Sigma$. Section \ref{sec:noncompact} contains
a proof of our main result that Abelian Chern-Simons theory does not admit a natural state.

\section{\label{sec:CStheory}Abelian Chern-Simons theory}
Let $M$ be a $3$-dimensional manifold which is diffeomorphic to a product
manifold $\bbR\times \Sigma$, where $\Sigma$ is a $2$-dimensional surface (possibly non-compact). 
Abelian Chern-Simons theory on $M$ is characterized by the action functional
\begin{flalign}\label{eqn:CSaction}
S_M^{} = \frac{1}{2} \int_M A\wedge \dd A~,
\end{flalign}
where $A\in \Omega^1(M)$ is the vector potential of a $U(1)$-connection.
The Euler-Lagrange equation corresponding to \eqref{eqn:CSaction} is given by
$\dd A =0$, i.e.\ the solution space of Abelian Chern-Simons theory
is the space of flat $U(1)$-connections $\Omega_{\dd}^1(M)$.
Taking the quotient by gauge transformations $A \mapsto A + \frac{1}{2\pi i}\, \dd\log (g)$, 
where $g\in C^\infty(M;U(1))$ is a circle-valued function, we obtain the gauge orbit space
$\Omega_{\dd}^1(M)/\Omega^1_{\bbZ}(M)$. Here $\Omega^1_{\bbZ}(M)$ denotes the
Abelian group of closed $1$-forms with integral periods.
By de Rham's theorem, this gauge orbit space is naturally isomorphic to the quotient
$H^1(M;\bbR)/ H^1(M;\bbZ)$ of the first cohomology group with real coefficients with respect 
to the first cohomology group with integer coefficients. 
Using homotopy invariance of cohomology groups,
we may contract away the real line $\bbR$ of the product manifold $M\simeq \bbR\times \Sigma$
and describe our gauge orbit space of interest from perspective of the surface $\Sigma$.
We shall use the notation
\begin{flalign}\label{eqn:modulispace}
\mathrm{Flat}_{U(1)}(\Sigma) := \frac{\Omega_{\dd}^1(\Sigma)}{\Omega^1_{\bbZ}(\Sigma)}~.
\end{flalign}
Let $\Man$ be the category of $2$-dimensional oriented manifolds (possibly non-compact) 
with morphisms given by orientation preserving open embeddings. Let $\Ab$ be the category of Abelian groups.
The gauge orbit spaces \eqref{eqn:modulispace} can be described by a functor
\begin{flalign}\label{eqn:modulifunctor}
\mathrm{Flat}_{U(1)}\, : \, \Man^\op\,  \longrightarrow \, \Ab\, ~,
\end{flalign}
which assigns to a morphism $f : \Sigma \to \Sigma^\prime$ in $\Man$ the corresponding pull-back
of differential forms, i.e.\
\begin{flalign}\label{eqn:modulimorphism}
\mathrm{Flat}_{U(1)}(f) := f^\ast\, :\,  \frac{\Omega_{\dd}^1(\Sigma^\prime)}{\Omega^1_{\bbZ}(\Sigma^\prime) }
\, \longrightarrow \, \frac{\Omega_{\dd}^1(\Sigma)}{\Omega^1_{\bbZ}(\Sigma)}\,~,~~ [A^\prime]\longmapsto [f^\ast A^\prime]~.
\end{flalign}

As basic observables for Abelian Chern-Simons theory we shall take all group characters
on $\mathrm{Flat}_{U(1)}(\Sigma)$, i.e.\ all group homomorphisms
$\mathrm{Flat}_{U(1)} (\Sigma) \to U(1)$ to the circle group.
The character group has a convenient description in terms of compactly supported differential forms
on $\Sigma$: Given any compactly supported $1$-form $\varphi \in \Omega^1_{\cc}(\Sigma)$,
we define a group character on $\Omega^1_{\dd}(\Sigma)$ by setting
\begin{flalign}\label{eqn:character}
\Omega^1_{\dd}(\Sigma) \longrightarrow U(1)~,~~A\longmapsto \exp\Big(2\pi i\,\int_{\Sigma} \varphi\wedge A\Big)~.
\end{flalign}
This character descends to the quotient \eqref{eqn:modulispace} if and only if
\begin{flalign}\label{eqn:integralcondition}
\int_{\Sigma}\varphi\wedge \Omega^1_{\bbZ}(\Sigma) \subseteq  \bbZ~.
\end{flalign}
Because $\dd\Omega^0(\Sigma)\subseteq \Omega^1_{\bbZ}(\Sigma)$ is a subgroup,
Stokes' lemma implies that any $\varphi\in\Omega^1_{\cc}(\Sigma)$ satisfying condition \eqref{eqn:integralcondition} has to be closed, i.e.\ $\varphi\in \Omega^1_{\cc,\,\dd}(\Sigma)$. 
In addition \eqref{eqn:integralcondition} entails an integrality condition 
\begin{flalign}\label{eqn:integralitycoho}
\ips{[\varphi]}{H^1(\Sigma;\bbZ)}^{}_\Sigma \subseteq \bbZ~
\end{flalign}
for the compactly supported cohomology class 
$[\varphi]\in H^1_{\cc}(\Sigma;\bbR):= \Omega^1_{\cc,\,\dd}(\Sigma)/\dd\Omega^0_\cc(\Sigma)$
of $\varphi$. Here 
\begin{flalign}
\ips{\,\cdot\,}{\,\cdot\,}_{\Sigma}^{} : H^1_{\cc}(\Sigma;\bbR)\times H^1(\Sigma;\bbR)\longrightarrow\bbR~,~~
([\varphi],[\omega])\longmapsto  \int_{\Sigma}\varphi\wedge \omega
\end{flalign}
is the non-degenerate pairing between compactly support and ordinary (de Rham) cohomology.
Because each exact $\varphi = \dd\chi\in\dd\Omega^0_{\cc}(\Sigma)$ yields a trivial group character \eqref{eqn:character},
it follows that the character group of $\mathrm{Flat}_{U(1)}(\Sigma)$ is isomorphic to
\begin{flalign}\label{eqn:charactergroup}
H^1_{\cc}(\Sigma;\bbZ) := \Big\{ [\varphi]\in H^1_{\cc}(\Sigma;\bbR) \, : \, \text{\eqref{eqn:integralitycoho} is satisfied}\Big\} ~.
\end{flalign}
The assignment of the character groups \eqref{eqn:charactergroup} is a functor
\begin{flalign}\label{eqn:characterfunctor}
H^1_{\cc}(-;\bbZ)\,:\, \Man \, \longrightarrow \, \Ab \,~,
\end{flalign}
which assigns to a morphism $f : \Sigma \to \Sigma^\prime$ in  $\Man$ the corresponding push-forward (extension by zero) 
of compactly supported differential forms, i.e.\
\begin{flalign}\label{eqn:charactermorphism}
H^1_{\cc}(f;\bbZ) := f_\ast \,:\,  H^1_{\cc}(\Sigma;\bbZ) \, \longrightarrow\,  H^1_{\cc}(\Sigma^\prime;\bbZ)\, ~,~~
[\varphi] \longmapsto [f_\ast\varphi]~.
\end{flalign}
The Abelian groups $H^1_{\cc}(\Sigma;\bbZ)$ can be equipped with a natural presymplectic structure
given by
\begin{flalign}\label{eqn:presymplectic}
\tau_{\Sigma}^{} : H^1_{\cc}(\Sigma;\bbZ)\times H^1_{\cc}(\Sigma;\bbZ) \longrightarrow \bbR~,~~([\varphi],[\widetilde{\varphi}])\longmapsto
\ips{[\varphi]}{[\widetilde{\varphi}]}_{\Sigma}^{}= \int_{\Sigma} \varphi\wedge \widetilde{\varphi}~.
\end{flalign}
(This presymplectic structure can be derived from the action functional \eqref{eqn:CSaction} 
by using Zuckerman's construction \cite{Zuckerman}.)
Given any $\Man$-morphism $f:\Sigma\to \Sigma^\prime$, we find that
\begin{flalign}
\tau_{\Sigma^\prime}^{} \big(f_\ast[\varphi],f_\ast[\widetilde{\varphi}]\big) = 
\int_{\Sigma^\prime} (f_\ast \varphi)\wedge (f_\ast \widetilde{\varphi}) = 
\int_{\Sigma} \varphi\wedge (f^\ast f_\ast \widetilde{\varphi}) = 
\int_{\Sigma} \varphi\wedge \widetilde{\varphi} = \tau_{\Sigma}^{}\big([\varphi],[\widetilde{\varphi}]\big) ~,
\end{flalign}
because $f$ preserves the orientations and $f^\ast f_\ast\widetilde{\varphi} = \widetilde{\varphi}$,
for all $\widetilde{\varphi}\in \Omega^1_{\cc}(\Sigma)$. Hence, \eqref{eqn:characterfunctor} can be
promoted to a functor 
\begin{flalign}\label{eqn:Ofunctor}
\OO := \big(H^1_{\cc}(-;\bbZ), \tau \big) \,:\, \Man \, \longrightarrow \,\mathsf{PAb}
\end{flalign}
with values in the category of presymplectic Abelian groups.
\sk

Quantization of Abelian Chern-Simons theory is achieved by composing the functor
\eqref{eqn:Ofunctor} with the canonical commutation relation functor
$\mathfrak{CCR} : \mathsf{PAb}\to \CAlg $ that assigns $C^\ast$-algebras to presymplectic Abelian groups,
see \cite{Manuceau} and for more details also the Appendix of \cite{BDHS}.
We shall denote the resulting functor by
\begin{flalign}\label{eqn:QFTfunctor}
\AA := \mathfrak{CCR}\circ \OO \,:\, \Man \longrightarrow \CAlg~.
\end{flalign}
Concretely, the $C^\ast$-algebra $\AA(\Sigma)$ is the $C^\ast$-completion of the $\ast$-algebra
generated by the Weyl symbols $W_{[\varphi]}$, for all $[\varphi]\in H^1_{\cc}(\Sigma;\bbZ)$,
satisfying the usual relations
\begin{flalign}\label{eqnWeylrel}
W_{[\varphi]}\, W_{[\widetilde{\varphi}]} := e^{-i \hbar \,\tau_{\Sigma}^{}([\varphi],[\widetilde{\varphi}])}~W_{[\varphi]+ [\widetilde{\varphi}]}\quad,\qquad {W_{[\varphi]}}^\ast := W_{-[\varphi]}~~.
\end{flalign}
We shall always assume that $\hbar \not\in 2\pi\bbZ$ to avoid trivial exponentials in \eqref{eqnWeylrel}, i.e.\ 
commutative $C^\ast$-algebras.
The $C^\ast$-algebra homomorphism $\AA(f) : \AA(\Sigma)\to\AA(\Sigma^\prime)$
is specified by
\begin{flalign}\label{eqn:explicitactionWeyl}
\AA(f) \big(W_{[\varphi]}\big) := W^\prime_{f_\ast[\varphi]}~~,
\end{flalign}
for all $[\varphi]\in H^1_{\cc}(\Sigma;\bbZ)$, where by $W^\prime$ we denote the Weyl symbols
in $\AA(\Sigma^\prime)$.

\section{\label{sec:compact}Invariant states on compact surfaces}
Any object $\Sigma$ in $\Man$ comes together
with its automorphism group, which is the group of orientation preserving diffeomorphisms
$\Diff^+(\Sigma)$ of $\Sigma$. The functor \eqref{eqn:QFTfunctor} 
defines a representation of $\Diff^+(\Sigma)$ on $\AA(\Sigma)$ in terms of 
$C^\ast$-algebra automorphisms, i.e.\
\begin{flalign}\label{eqn:diffeorepresentation}
\Diff^+(\Sigma)\longrightarrow \mathrm{Aut}(\AA(\Sigma))~,~~f \longmapsto \AA(f) ~.
\end{flalign}
Because of \eqref{eqn:explicitactionWeyl} and of the fact that $H^1_{\cc}(\Sigma;\bbZ)$ is discrete,
the identity component $\Diff^+_0(\Sigma)\subseteq \Diff^+(\Sigma)$ 
is represented trivially. Hence, the representation \eqref{eqn:diffeorepresentation} descends
to the mapping class group
\begin{flalign}\label{eqn:MCGrepresentation}
\mathrm{MCG}(\Sigma) := \frac{\Diff^+(\Sigma)}{\Diff^+_0(\Sigma)} 
\longrightarrow \mathrm{Aut}(\AA(\Sigma))~,~~[f]\longmapsto \AA(f)~.
\end{flalign}
For compact $\Sigma$, there exists by \cite[Chapter 6]{MCG} a short exact sequence of groups
\begin{flalign}
\xymatrix{
1 ~\ar[r] & ~\mathrm{Tor}(\Sigma) ~\ar[r] & ~\mathrm{MCG}(\Sigma)~ \ar[r] & ~\mathrm{Sp}(H^1(\Sigma;\bbZ),\tau_{\Sigma})~\ar[r]&~1
}~~,
\end{flalign}
where $\mathrm{Tor}(\Sigma)$ is the so-called Torelli group and $\mathrm{Sp}(H^1(\Sigma;\bbZ),\tau_{\Sigma})$ 
is the symplectic group. (Notice that $H_{\cc}^1(\Sigma;\bbZ) = H^1(\Sigma;\bbZ)$ for compact $\Sigma$.)
Because the Torelli group is by definition the kernel of the representation 
of the mapping class group on homology, the representation \eqref{eqn:MCGrepresentation} 
descends further to a representation of $\mathrm{Sp}(H^1(\Sigma;\bbZ),\tau_{\Sigma})$, 
for all compact $\Sigma$. Explicitly,
\begin{subequations}\label{eqn:Taction}
\begin{flalign}
\mathrm{Sp}(H^1(\Sigma;\bbZ),\tau_{\Sigma}) \longrightarrow \mathrm{Aut}(\AA(\Sigma))~,~~T\longmapsto \kappa_T~,
\end{flalign}
where the $C^\ast$-algebra automorphism $\kappa_T$ is specified by
\begin{flalign}
\kappa_{T}^{} : \AA(\Sigma) \longrightarrow \AA(\Sigma)~,~~W_{[\varphi]}\longmapsto W_{T[\varphi]}~,
\end{flalign}
\end{subequations}
for all $T\in \mathrm{Sp}(H^1(\Sigma;\bbZ),\tau_{\Sigma})$.
\sk

As a consequence, we obtain that, for compact $\Sigma$,
a state $\omega : \AA(\Sigma) \longrightarrow \bbC$
is invariant under the 
action \eqref{eqn:diffeorepresentation} of the orientation 
preserving diffeomorphism group $\Diff^+(\Sigma)$ if and only if it satisfies
\begin{flalign}
\omega\circ \kappa_T = \omega~,
\end{flalign}
for all $T\in  \mathrm{Sp}(H^1(\Sigma;\bbZ),\tau_{\Sigma})$.
Notice that there exists an invariant state given by
\begin{flalign}\label{eqn:invariantstate}
\omega\big(W_{[\varphi]}\big) := \begin{cases}
1 & ~,~~\text{if }[\varphi]=0~,\\
0 & ~,~~\text{else}~,
\end{cases}
\end{flalign}
for any compact $\Sigma$.
\begin{rem}
It would be interesting to find out whether \eqref{eqn:invariantstate} is the only 
$\Diff^+(\Sigma)$-invariant state on $\AA(\Sigma)$, for $\Sigma$ compact.
Unfortunately, it seems to be rather hard to answer this question in full generality. As a partial result
in this direction, we provide an argument why there {\em does not} exist an invariant Gaussian state 
on $\AA(\Sigma)$. Recall that a Gaussian state is of the form
\begin{flalign}
\omega\big(W_{[\varphi]}\big) = e^{- \mu([\varphi],[\varphi])}~,
\end{flalign}
for all $[\varphi]\in H^1(M;\bbZ)$, where
\begin{flalign}
\mu : H^1(M;\bbZ)\times H^1(M;\bbZ)\longrightarrow \bbR
\end{flalign}
is a symmetric and positive-definite bilinear form satisfying a 
lower bound which depends on the symplectic structure $\tau_{\Sigma}$, see e.g.\ \cite{Manuceau2}.
Such state is invariant under $\Diff^+(\Sigma)$ if and only if
$\mu$ is invariant under  $\mathrm{Sp}(H^1(\Sigma;\bbZ),\tau_{\Sigma})$.
This is however impossible because $\mathrm{Sp}(H^1(\Sigma;\bbZ),\tau_{\Sigma})$
is not a subgroup of the orthogonal group  $\mathrm{O}(H^1(\Sigma;\bbZ),\mu)$.
\end{rem}

\section{\label{sec:noncompact}Non-existence of a natural state}
A  {\em natural state} for the Abelian Chern-Simons theory
functor $\AA :\Man \to \CAlg$ is a state $\omega_{\Sigma}^{} : \AA(\Sigma)\to \bbC$ 
for each object $\Sigma$ in $\Man$,
such that for all $\Man$-morphisms $f : \Sigma\to \Sigma^\prime$
the consistency condition
\begin{flalign}
\omega_{\Sigma^\prime}^{} \circ \AA(f) = \omega_{\Sigma}^{}~
\end{flalign}
holds true. In particular, this condition implies that for each
$\Sigma$ the state $\omega_{\Sigma}^{}$ has 
to be invariant under the action \eqref{eqn:diffeorepresentation}
of the orientation preserving diffeomorphism group of $\Sigma$.
\begin{theo}\label{theo:nonaturalstate}
There exists no natural state for the Abelian Chern-Simons theory functor $\AA :\Man \to \CAlg$.
\end{theo}
\begin{proof}
Let us assume that there exists a natural state $\{\omega_{\Sigma}\}_{\Sigma\in\Man}$. Consider the
$\Man$-diagram 
\begin{flalign}\label{eqn:mapping}
\bbS^2 \stackrel{f_1}{\longleftarrow} \bbR\times \bbT\stackrel{f_2}{\longrightarrow} \bbT^2
\end{flalign}
describing an orientation preserving open embedding of the cylinder $\bbR\times\bbT$ into the $2$-sphere $\bbS^2$ and 
the $2$-torus $\bbT^2$. The Chern-Simons functor $\AA :\Man \to \CAlg$ assigns
the $\CAlg$-diagram $\AA(\bbS^2) \stackrel{\AA(f_1)}{\longleftarrow} \AA(\bbR\times \bbT)
\stackrel{\AA(f_2)}{\longrightarrow} \AA(\bbT^2)$ and naturality of the state implies the condition
\begin{flalign}\label{eqn:statecondition}
\omega_{\bbS^2}^{}\circ \AA(f_1) = \omega_{\bbR\times\bbT}^{} = \omega_{\bbT^2}^{}\circ \AA(f_2)~.
\end{flalign}
Because of $H^1_{\cc}(\bbS^2;\bbZ)=0$, it follows that $\AA(\bbS^2) \simeq \bbC$
and hence $\omega_{\bbS^2}^{} = \id_{\bbC}$ has to be the unique state on $\bbC$.
Using further that $H^1_\cc(\bbR\times \bbT) \simeq \bbZ$, it follows
that $\AA(\bbR\times\bbT)\simeq \mathfrak{CCR}(\bbZ,0)$, and the first equality 
in \eqref{eqn:statecondition} implies
\begin{flalign}\label{eqn:condition1}
\omega_{\bbR\times\bbT}^{}\big(W^{\bbR\times\bbT}_n\big)= 1~,
\end{flalign}
for all $n\in\bbZ$. (The superscript on the Weyl symbols refers to the algebras they live in.)
\sk

It is easy to show that $(H^1_{\cc}(\bbT^2;\bbZ),\tau_{\bbT^2})$ is isomorphic to the Abelian group
$\bbZ^2$ with standard symplectic structure
\begin{flalign}\label{eqn:standardsymp}
\bbZ^2 \times \bbZ^2 \longrightarrow \bbR~,~~\bol{n} \times \bol{m}\longmapsto n_1\,m_2 - n_2\,m_1~,
\end{flalign}
where we used the notation $\bol{n}=(n_1,n_2)\in\bbZ^2$ and similarly for $\bol{m}$.
We can choose $f_2$ such that the 
$C^\ast$-algebra homomorphism $\AA(f_2) : \AA(\bbR\times\bbT)\to\AA(\bbT^2)$
is given by $W^{\bbR\times\bbT}_n \mapsto W^{\bbT^2}_{(n,0)}$, for all $n\in \bbZ$.
As a consequence of \eqref{eqn:statecondition} and \eqref{eqn:condition1}, we obtain that
\begin{flalign}
\omega_{\bbT^2}^{}\big(W^{\bbT^2}_{(n,0)}\big) =1~,
\end{flalign}
for all $n\in\bbZ$. Even more, because $\omega_{\bbT^2}^{}$ is by hypothesis
a component of a natural state and hence invariant under the action \eqref{eqn:diffeorepresentation} 
orientation preserving diffeomorphisms of $\bbT^2$, we obtain that
\begin{flalign}\label{eqn:condition2}
\omega_{\bbT^2}^{}\big(W^{\bbT^2}_{T (n,0)}\big) =1~,
\end{flalign}
for all $n\in \bbZ$ and all $T\in \mathrm{SL}(2,\bbZ) \simeq \mathrm{Sp}(H^1_{\cc}(\bbT^2;\bbZ),\tau_{\bbT^2})$, 
see also \eqref{eqn:Taction}.
\sk

We conclude the proof by showing that the necessary conditions \eqref{eqn:condition2}
are incompatible with positivity of $\omega_{\bbT^2}^{} : \AA(\bbT^2)\to\bbC$.
For this we consider the particular element 
$a = \alpha_1\,\1 + \alpha_2 \,W^{\bbT^2}_{(1,1)} + \alpha_3\, W^{\bbT^2}_{(0,1)} \in\AA(\bbT^2)$,
where $\alpha_1,\alpha_2,\alpha_3\in\bbC$ are arbitrary. Notice 
that $(1,1)$, $(0,1)$ and their difference $(1,0)$ are in the $\mathrm{SL}(2,\bbZ)$-orbit
of $\bbZ\times 0\subseteq \bbZ^2$, which allows us to use \eqref{eqn:condition2}.
Using also the Weyl relations \eqref{eqnWeylrel} and \eqref{eqn:standardsymp},
we can evaluate $\omega_{\bbT^2}^{}(a^\ast\,a)$ and obtain
\begin{flalign}\label{eqn:matrix}
\omega_{\bbT^2}^{}\big(a^\ast\,a\big) = \begin{pmatrix} \overline{\alpha_1}&\overline{\alpha_2}&\overline{\alpha_3}\end{pmatrix}
\begin{pmatrix}
1&1&1\\
1&1&e^{i\hbar}\\
1 & e^{-i\hbar} &1
\end{pmatrix}
\begin{pmatrix}
\alpha_1\\ \alpha_2\\ \alpha_3
\end{pmatrix}~~,
\end{flalign}
where $\overline{\,\cdot\,}$ denotes complex conjugation. Because $\hbar\not\in2\pi\bbZ$ (by assumption),
the matrix in \eqref{eqn:matrix} has a negative eigenvalue, which implies that $\omega_{\bbT^2}^{}$ is not positive
and hence not a state.
\end{proof}
\begin{rem}
Note that our argument in the proof of Theorem \ref{theo:nonaturalstate} relies
on the fact that the Chern-Simons functor $\AA : \Man \to \CAlg$ violates
the isotony axiom of locally covariant quantum field theory, i.e.\ there exist
$\Man$-morphisms $f : \Sigma\to\Sigma^\prime$ such that $\AA(f) : \AA(\Sigma)\to\AA(\Sigma^\prime)$
is not injective. An example is given by the map $f_1$ in \eqref{eqn:mapping}.
Violation of isotony seems to be a general feature of quantum gauge theories (see e.g.\ \cite{BDHS} for the
case of Abelian Yang-Mills theory), including our present example of Abelian Chern-Simons theory.
As all topological quantum field theories known to us are gauge theories,
we expect that similar arguments based on the violation of isotony could be used to develop
a more general argument why topological quantum field theories do not admit natural states.
We hope to come back to this issue in a future work.
\end{rem}

\section*{Acknowledgements}
We would like to thank Ulrich Bunke, Giuseppe De Nittis, Chris Fewster 
and Nicola Pinamonti for useful comments on this work. C.D.\ is supported by the University 
of Pavia and also would like to thank Felix Finster for the invitation to Regensburg, where 
part of this research has been performed.  S.M.\ is supported by the DFG Research Training Group 
GRK 1692 ``Curvature, Cycles, and Cohomology'' and also would like to thank Giuseppe 
De Nittis for the invitation to Santiago de Chile, where part of this research has been 
performed. A.S.\ gratefully acknowledges the financial support of the Royal Society 
(UK) through a Royal Society University Research Fellowship.

\end{document}